\pdfoutput=1
\documentclass[11pt]{article}

\usepackage[utf8]{inputenc}
\DeclareUnicodeCharacter{211D}{\ensuremath{\mathbb{R}}}   
\DeclareUnicodeCharacter{03C6}{\ensuremath{\varphi}}      
\DeclareUnicodeCharacter{00A7}{\S}
\DeclareUnicodeCharacter{2014}{\textemdash}
\DeclareUnicodeCharacter{2265}{\ensuremath{\geq}}        
\DeclareUnicodeCharacter{2264}{\ensuremath{\leq}}        
\DeclareUnicodeCharacter{2192}{\ensuremath{\to}}         
\DeclareUnicodeCharacter{2200}{\ensuremath{\forall}}     
\DeclareUnicodeCharacter{2203}{\ensuremath{\exists}}     
\DeclareUnicodeCharacter{2208}{\ensuremath{\in}}         
\DeclareUnicodeCharacter{2227}{\ensuremath{\wedge}}      
\DeclareUnicodeCharacter{222B}{\ensuremath{\int}}        
\DeclareUnicodeCharacter{2202}{\ensuremath{\partial}}    
\DeclareUnicodeCharacter{222A}{\ensuremath{\cup}}        
\DeclareUnicodeCharacter{03BC}{\ensuremath{\mu}}         
\DeclareUnicodeCharacter{03C9}{\ensuremath{\omega}}      
\DeclareUnicodeCharacter{03A9}{\ensuremath{\Omega}}      
\DeclareUnicodeCharacter{2016}{\ensuremath{\Vert}}       
\DeclareUnicodeCharacter{00D7}{\ensuremath{\times}}      
\DeclareUnicodeCharacter{02E2}{\ensuremath{^{s}}}        
\DeclareUnicodeCharacter{1D50}{\ensuremath{^{m}}}        
\DeclareUnicodeCharacter{2080}{\ensuremath{_0}}          
\DeclareUnicodeCharacter{2081}{\ensuremath{_1}}          
\DeclareUnicodeCharacter{2082}{\ensuremath{_2}}          
\usepackage[T1]{fontenc}
\usepackage[sc]{mathpazo}         
\usepackage[a4paper,margin=1in]{geometry}
\usepackage{amsmath}
\usepackage{amssymb}
\usepackage{amsthm}
\usepackage{booktabs}
\usepackage{array}
\usepackage{microtype}
\usepackage{xcolor}
\definecolor{leanbg}{HTML}{F7F5F0}
\definecolor{rulecol}{HTML}{8A7A5C}
\usepackage[most]{tcolorbox}
\usepackage{titlesec}
\titleformat{\section}{\large\scshape}{\thesection}{0.8em}{}
\titleformat{\subsection}{\normalsize\scshape}{\thesubsection}{0.7em}{}
\usepackage{alltt}
\usepackage{tikz}
\usetikzlibrary{arrows.meta}
\usepackage[hidelinks]{hyperref}

\newcommand{\lean}[1]{\texttt{\def\_{\textunderscore\allowbreak}#1}}
\newcommand{\leanw}[1]{\texttt{#1}}
\newcommand{\E}{\mathbb{E}}
\newcommand{\R}{\mathbb{R}}
\newcommand{\F}{\mathcal{F}}
\newcommand{\ItoI}{\mathcal{I}}

\newtcolorbox{leanbox}{enhanced, sharp corners, boxrule=0.5pt,
  colback=leanbg, colframe=rulecol, left=8pt, right=8pt, top=6pt, bottom=6pt}
\newcounter{leanlst}
\newenvironment{leanlisting}[2]%
  {\par\medskip\refstepcounter{leanlst}\label{#2}%
   \noindent{\small\itshape Listing \theleanlst.\ #1}\par\nobreak\smallskip
   \begin{leanbox}\footnotesize\begin{alltt}}
  {\end{alltt}\end{leanbox}\par\medskip\noindent\ignorespaces}

\newtheoremstyle{refined}{\topsep}{\topsep}{\normalfont}{}{\scshape}{.}{.6em}{}
\theoremstyle{refined}
\newtheorem{theorem}{Theorem}[section]

\theoremstyle{remark}

\title{A Machine-Checked It\^o Calculus for Brownian Motion}
\author{Raphael Coelho\thanks{Independent researcher.
  ORCID: \href{https://orcid.org/0009-0001-6601-1023}{0009-0001-6601-1023}.
  Correspondence: \texttt{raphaelrrcoelho@gmail.com}.
  Artifact: \url{https://github.com/raphaelrrcoelho/formal-mathfin}.}}
\date{June 2026}

\begin{document}
\maketitle

\begin{abstract}
We develop the It\^o calculus of Brownian motion, machine-checked in Lean~4
over Mathlib and the \lean{BrownianMotion} package. On a bounded interval
$[0,T]$ the It\^o integral is built as a Hilbert-space isometry, from a
predictable-rectangle $\pi$-system through the density of simple adapted
processes. Realized as a process, it is a continuous $L^2$ martingale. One
structural identity drives this: the integral at time $t$ is the
conditional-expectation projection of its terminal value onto $\F_t$, and
from it adaptedness, the martingale property, the contraction bound, and both
the terminal and time-indexed It\^o isometries follow as corollaries. On this
integral we prove It\^o's formula for $C^3$ functions with bounded
derivatives, including the time-dependent form
$df = f_x\,dB + (f_t + \tfrac12 f_{xx})\,dt$, by a discrete-to-continuous
argument through weighted quadratic variation with explicit $L^2$ remainder
bounds. We then pass from the $L^2$ theory to the pathwise. The integral
process has an almost-surely continuous modification, and its
everywhere-continuous representative is a local martingale for the
null-augmented Brownian filtration; gluing the bounded-horizon
representatives along the half-line yields the It\^o integral as a continuous
local martingale on all of $\R_{\ge 0}$, the form it takes in the classical
theory. To our knowledge these are the first machine-checked constructions of
the It\^o integral and of It\^o's formula in any proof assistant, and the
first to reach a pathwise-continuous local martingale. The boundary is
explicit. The $L^2$ integral and It\^o's formula are developed on $[0,T]$
with bounded-derivative integrands; the unrestricted $C^2$ formula,
integrators beyond Brownian motion, and right-continuity of the filtration
lie outside the development. It is roughly 7{,}900 lines of Lean across 26
modules, every theorem \lean{sorry}-free, the axioms of each headline result
pinned to Mathlib's classical defaults by a build-enforced gate, and the
whole reproducible from a pinned toolchain.
\end{abstract}

\section{Introduction}\label{sec:intro}

Stochastic calculus sits oddly in the formalization landscape. Its
prerequisites (measure-theoretic probability, conditional expectation,
$L^p$ spaces, Gaussian measures) are by now well represented in the
major proof-assistant libraries. Its consumers (mathematical finance,
stochastic control, SPDEs) are fields with strong correctness cultures
that regularly cite formal methods as an aspiration. Yet the calculus
itself has remained unformalized. Martingales in discrete and continuous
time were formalized in Isabelle/HOL by Keskin~\cite{keskin2023}; the
central limit theorem was machine-checked a decade ago~\cite{avigad2017};
and in late 2025 Degenne, Ledvinka, Marion, and Pfaffelhuber completed the
first formal construction of Brownian motion itself, in
Lean~4~\cite{brownianMotionRepo,degenne2025}. Their construction is the
foundation the present work builds on; the It\^o integral, It\^o's formula,
and the martingale theory of the integral are the next layer up. As of this
writing we are not aware of any machine-checked proof of It\^o's formula, or
any machine-checked construction of the It\^o integral, in any proof
assistant.

This paper describes such a development. It is the stochastic-calculus
core of a larger formalized mathematical-finance library announced
in~\cite{flagship}; where the announcement surveys breadth, the present
paper is a detailed treatment of its stochastic-calculus core, for readers
who want to know exactly what was proved, under which hypotheses, and by
which proof architecture.

\subsection*{Contributions}

All results are stated over a probability space $(\Omega,\F,\mu)$ carrying
a process $B : \R_{\ge 0} \to \Omega \to \R$ satisfying the
\lean{BrownianMotion} package's pre-Brownian interface
(\lean{IsPreBrownian}): centered Gaussian increments with variance
$t-s$, measurable coordinates, and, where stated, continuous paths.
$\F_t$ denotes the natural filtration (\lean{natFiltration}).

\begin{enumerate}
\item \textbf{The It\^o integral as an isometry}
  (Section~\ref{sec:integral}). The integrand space is the $L^2$ space of
  a product measure on $(0,T]\times\Omega$ trimmed to the predictable
  $\sigma$-algebra generated by rectangles $(s,t]\times A$, $A \in \F_s$.
  Simple adapted processes are dense; the elementary integral is an
  isometry on them; the It\^o integral $\ItoI_T = $
  \lean{itoIntegralCLM\_T} is its unique continuous-linear isometric
  extension.
\item \textbf{The integral as a process}
  (Section~\ref{sec:process}). For each $t \le T$ a continuous linear map
  realizes $(\varphi \bullet B)_t = \int_0^t \varphi\,dB$, and the
  development's load-bearing identity states
  $(\varphi \bullet B)_t = \E[\,\ItoI_T(\varphi) \mid \F_t\,]$ as an
  identity of $L^2$ elements. Adaptedness, the martingale property, the
  contraction bound, the terminal isometry, the time-indexed It\^o
  isometry $\E[(\varphi\bullet B)_t^2] = \int_0^t \E[\varphi_s^2]\,ds$, and
  $L^2$-continuity in $t$ are corollaries.
\item \textbf{It\^o's formula} (Section~\ref{sec:formula}). For $f \in
  C^3(\R)$ with $|f'| \le C_1$, $|f''| \le C_2$, $|f'''| \le C_3$,
  \[
    f(B_T) - f(B_0) \;=\; \ItoI_T\big(f'(B)\big) \;+\; \tfrac12
    \int_0^T f''(B_s)\,ds \qquad \mu\text{-a.e.},
  \]
  and the time-dependent form
  $f(T, B_T) - f(0, B_0) = \ItoI_T(f_x(\cdot, B)) + \int_0^T (f_t +
  \tfrac12 f_{xx})(s, B_s)\,ds$ under a $C^{1,2}$-with-bounds hypothesis
  package.
\item \textbf{The integral as a pathwise local martingale}
  (Section~\ref{sec:pathwise}). The process has an almost-surely continuous
  modification on $[0,T]$
  (\lean{exists\_continuous\_modification\_itoProcess}); its
  everywhere-continuous representative is a local martingale for the
  null-augmented Brownian filtration
  (\lean{exists\_continuous\_localMartingale\_modification}); and gluing the
  bounded-horizon representatives along the half-line gives the It\^o
  integral as a continuous local martingale on all of $\R_{\ge 0}$
  (\lean{exists\_continuous\_localMartingale\_modification\_infinite}).
\item \textbf{Supporting theory}. Quadratic variation of Brownian motion
  in $L^1$ and $L^2$ forms under hypotheses strictly weaker than the
  textbook statement; the expectation-level It\^o identity $\E[f(B_t)] =
  f(0) + \tfrac12 \int_0^t \E[f''(B_s)]\,ds$; and the closed-form keystone
  $\ItoI_T(B) = \tfrac12(B_T^2 - B_0^2 - T)$, the formal counterpart of
  the classical $\int_0^T B\,dB$.
\end{enumerate}

Section~\ref{sec:scope} sets the boundary precisely: which parts of the
classical theory are present, which are absent, and why the title says
``a'' calculus rather than ``the'' calculus.

\section{Setting and verification discipline}\label{sec:setting}

\subsection{Foundations consumed}

The development is built in Lean~4~\cite{lean4} over
Mathlib~\cite{mathlib} and the \lean{BrownianMotion}
package~\cite{brownianMotionRepo,degenne2025}, whose construction of
Brownian motion (via Kolmogorov--Chentsov, Gaussian measures on Banach
spaces, and a Carath\'eodory/Kolmogorov extension layer) this work consumes
rather than reproves; the artifact pins the package at commit
\leanw{eaa4391}. Time is indexed by $\R_{\ge 0}$ (Lean's \lean{NNReal});
Lebesgue measure on time enters through a dedicated \lean{timeMeasure}
restricted to $(0,T]$.

\subsection{Auditability and reproducibility}

Verification discipline follows the library's standing
contract~\cite{flagship}. No \lean{sorry} or \lean{admit} occurs anywhere
in the development. A build-enforced axiom audit pins the axioms of every
headline constant to Mathlib's classical defaults (\lean{propext},
\lean{Classical.choice}, \lean{Quot.sound}), and a continuous-integration
gate fails the build if a new proof introduces any other axiom. A hash
ledger records the exact source inputs each verified statement depends on,
so that a result cannot become stale without detection when an upstream
dependency changes. The
statements displayed below are transcriptions of the Lean statements; where
a hypothesis package is abbreviated in prose, the displayed Lean name is
the authority, and Appendix~\ref{app:names} maps every displayed result to
its module and declaration name.

\subsection{Scale}\label{sec:stats}

The It\^o and quadratic-variation tower is roughly $7{,}900$ lines of Lean
across $26$ modules, comprising on the order of $250$ theorems and lemmas
and $40$ definitions, all downstream of Mathlib and the
\lean{BrownianMotion} package and upstream of the parent library's pricing
layers. Table~\ref{tab:modules} groups the modules by role. These figures
are reported from the working tree of June~2026 and are refreshed against a
clean build before submission.

\begin{table}[t]
\centering
\small
\begin{tabular}{@{}lp{8.4cm}@{}}
\toprule
Role & Modules (abbreviated) \\
\midrule
Discrete engine &
  \lean{DiscreteIto}, \lean{DiscreteItoPolynomial},
  \lean{ItoSquaringIdentity} \\
Quadratic variation &
  \lean{QuadraticVariationL2}, \lean{WeightedQuadraticVariation},
  \lean{BrownianQuadraticVariation} \\
Integral (isometry) &
  \lean{ItoIntegralL2}, \lean{ItoIntegralCLM},
  \lean{ItoIsometryAdapted}, \lean{ItoIntegralBrownian} \\
Riemann bridge &
  \lean{ItoIntegralRiemannBridge}, \lean{ItoIntegralRiemannBridgeTD} \\
Integral as process &
  \lean{ItoIntegralProcess}, \lean{ItoIntegralProcessGeneral},
  \lean{ItoIntegralProcessMartingale},
  \lean{ItoIntegralProcessIsometry} \\
It\^o's formula &
  \lean{ItoFormulaC2}, \lean{ItoFormulaCLM},
  \lean{ItoFormulaRemainder}, \lean{ItoFormulaSquaredL2},
  \lean{ItoFormulaTD}, \lean{ItoFormulaTDRemainder} \\
Pathwise regularity &
  \lean{ItoIntegralProcessContinuousModification},
  \lean{ItoIntegralProcessLocalMartingaleGeneral},
  \lean{ItoIntegralProcessL2Infinite},
  \lean{ItoIntegralProcessLocalMartingaleInfinite} \\
\bottomrule
\end{tabular}
\caption{The development by role. Roughly $7{,}900$ lines total.}
\label{tab:modules}
\end{table}

\section{The classical theory}\label{sec:classical}

We recall the objects the formalization reproduces, in the form they take in
the standard references~\cite{karatzas1991,ito1944}, so that the theorems of
Sections~\ref{sec:integral} to~\ref{sec:formula} can be read against them.

Let $(B_t)_{t \ge 0}$ be a Brownian motion on $(\Omega,\F,\mu)$ with natural
filtration $(\F_t)$. For a simple adapted process $\varphi_s = \sum_k \xi_k
\mathbf{1}_{(t_k, t_{k+1}]}(s)$, with $\xi_k$ bounded and
$\F_{t_k}$-measurable, the It\^o integral is the finite sum $\int_0^T
\varphi\,dB = \sum_k \xi_k (B_{t_{k+1}} - B_{t_k})$. The defining computation
is the \emph{It\^o isometry}: because each $\xi_k$ is independent of the
forward increment $B_{t_{k+1}} - B_{t_k}$, and increments over disjoint
intervals are independent,
\[
  \E\Big[\Big(\textstyle\sum_k \xi_k (B_{t_{k+1}} - B_{t_k})\Big)^2\Big]
  = \sum_k \E[\xi_k^2]\,(t_{k+1} - t_k) = \E\!\int_0^T \varphi_s^2 \, ds,
\]
the cross terms vanishing by adaptedness and the diagonal contributing
$\E[\xi_k^2]\,(t_{k+1}-t_k)$ because $\E[(B_{t_{k+1}}-B_{t_k})^2 \mid \F_{t_k}]
= t_{k+1}-t_k$. The map $\varphi \mapsto \int_0^T \varphi\,dB$ is thus an
isometry from the simple processes, carrying the $L^2(\lambda \otimes \mu)$
norm on the predictable side, into $L^2(\mu)$; as the simple processes are
dense in the predictable $L^2$ space, it extends to a linear isometry on all
of it. That construction is Section~\ref{sec:integral}.

As a function of its upper limit, $t \mapsto \int_0^t \varphi\,dB$ is a
continuous $L^2$ martingale, with $\E\big[(\int_0^t \varphi\,dB)^2\big] =
\E\int_0^t \varphi_s^2\,ds$ at every $t$; that is Section~\ref{sec:process}.
The quadratic variation of Brownian motion is $[B]_t = t$: along partitions of
$[0,t]$ with mesh tending to zero, $\sum_k (B_{t_{k+1}} - B_{t_k})^2 \to t$,
in $L^2$ for the uniform partition. It\^o's formula states that for $f \in
C^2$,
\[
  f(B_t) - f(B_0) = \int_0^t f'(B_s)\,dB_s + \tfrac12 \int_0^t f''(B_s)\,ds,
\]
the second-order term surviving precisely because $(\mathrm dB)^2 = \mathrm dt$
in the mean-square sense made exact by $[B]_t = t$. Section~\ref{sec:formula}
formalizes this for $f$ with bounded derivatives, and its time-dependent
extension.

\section{Design of the formalization}\label{sec:design}

Three representation choices shape the development and are worth stating before
the theorems that rest on them.

\paragraph{The integrand, after the upstream package.} The simple integrands
are not re-encoded here. The \lean{BrownianMotion} package provides
\lean{SimpleProcess}, an integrand built as a finite linear combination
(\lean{Finsupp.linearCombination}) of indicators of elementary predictable
sets, together with the lemma identifying the $\sigma$-algebra those indicators
generate with the predictable one. This work consumes both: the elementary
It\^o integral is defined on \lean{SimpleProcess}, and the density argument of
Section~\ref{sec:integral} rests on that generation lemma. What is built here
is the layer above, namely the rectangle $\pi$-system, the
orthogonal-complement density, and the isometric extension.

\paragraph{Time as \lean{NNReal}.} The time axis is $\R_{\ge 0}$, not $\R$ with
a positivity side-condition. Increments $B_t - B_s$ then carry $s \le t$ in
their type-level data rather than as a hypothesis to discharge, and the bounded
horizon enters uniformly through a single finite measure (\lean{timeMeasure\_T},
of total mass $T$) rather than through interval hypotheses threaded across
lemmas.

\paragraph{One construction pattern, used twice.} The integral and the
integral-as-a-process are built by the same idiom: a norm-preserving linear map
on a dense subspace, then its unique continuous-linear extension (Mathlib's
\lean{extendOfNorm}). For the integral the dense subspace is the simple
processes inside $H_T$; for the process it is the same simple processes, the
extension now landing in $L^2(\mu)$ at each $t$. Because the process is an
extension of the simple-layer map, every identity proved on simple processes
(the martingale property, the per-time isometry) transfers to all of $H_T$ by
continuity and density. That mechanism is what makes Section~\ref{sec:process}
a sequence of short arguments.

\section{The $L^2$ It\^o integral}\label{sec:integral}

\subsection{The integrand space}

Fix $T \in \R_{\ge 0}$ and write $\lambda_T$ for Lebesgue measure on time
restricted to $(0,T]$ (\lean{timeMeasure\_T}). On $(0,T] \times \Omega$ the
predictable rectangles
\[
  \mathcal{R} \;=\; \big\{ (s,t] \times A \;:\; s \le t,\; A \in \F_s
  \big\} \;\cup\; \big\{ \{0\}\times A_0 : A_0 \in \F_0 \big\}
\]
form a $\pi$-system, encoded as \lean{predictableRect}
(Listing~\ref{lst:rect}); the $\sigma$-algebra they generate is the
(restricted) predictable $\sigma$-algebra, identified by
\lean{generateFrom\_predictableRect}. The integrand space is
\[
  H_T \;=\; L^2\big((0,T]\times\Omega,\; \mathcal{P}_T,\;
  (\lambda_T \otimes \mu)\!\restriction_{\mathcal{P}_T}\big),
\]
realized in Lean as the $L^2$ space of the product measure \emph{trimmed}
to the predictable $\sigma$-algebra (\lean{trimMeasure\_T}); the It\^o
integral is then a continuous linear map out of this space
(Listing~\ref{lst:clm}).

\begin{leanlisting}{The predictable-rectangle $\pi$-system (abbreviated).}{lst:rect}
def predictableRect (hBmeas : ∀ t, Measurable (B t)) :
    Set (Set (ℝ≥0 × Ω)) :=
  -- the {0} × F₀ piece
  {S | ∃ F₀, MeasurableSet[natFiltration hBmeas 0] F₀ ∧ S = {0} ×ˢ F₀}
  ∪ -- the (a,b] × F piece, F measurable at time a
  {S | ∃ a b F, a ≤ b ∧ MeasurableSet[natFiltration hBmeas a] F ∧ ... }
\end{leanlisting}

\begin{leanlisting}{The It\^o integral as a continuous linear map out of the trimmed-predictable $L^2$ space.}{lst:clm}
noncomputable def itoIntegralCLM_T (T : ℝ≥0) (hBmeas : ∀ t, Measurable (B t)) :
    Lp ℝ 2 (trimMeasure_T T hBmeas) →L[ℝ] Lp ℝ 2 μ
\end{leanlisting}

\paragraph{Why a trimmed measure.} The one definitional choice that pays
off repeatedly is to make $H_T$ the $L^2$ space of a measure that has
\emph{already} been restricted to the predictable $\sigma$-algebra, rather
than carrying a sub-$\sigma$-algebra measurability side-condition alongside
an ambient $L^2$ space. The consequence is that predictability is a
property of \emph{membership in $H_T$}, not a separate hypothesis to be
threaded through every limit: a Cauchy sequence in $H_T$ has its limit in
$H_T$, so when It\^o's formula realizes $s \mapsto f'(B_s)$ as an $L^2$
limit of step processes (Section~\ref{sec:formula}), the limit's
predictability is free. The cost is a one-time investment in the
trimmed-measure $L^2$ API; the saving recurs at every density and
convergence argument downstream.

\subsection{Simple processes, density, isometry}

A simple adapted process is a finite combination of indicator rectangles
$\xi_s \cdot \mathbf{1}_{(s,t]}$ with $\xi_s$ bounded and
$\F_s$-measurable; \lean{TBoundedSP} packages the class and
\lean{simpleProcessL2\_T} sends it into $H_T$. The elementary integral maps
such a process to $\sum \xi_s\,(B_t - B_s)$ (\lean{itoAssembly\_T});
independence and the Gaussian increment moments give the elementary It\^o
isometry (\lean{assembly\_isometry\_T}),
$\big\|\text{(elementary integral of } V)\big\|_{L^2(\mu)} = \|V\|_{H_T}$.

Density of the simple processes in $H_T$
(\lean{simpleAssembly\_T\_denseRange}) is the technical core of the
construction, and the one part of it not inherited from the upstream package.
It is proved by the orthogonal-complement method. The basic predictable
rectangles $(a,b] \times F$ with $F \in \F_a$ form a $\pi$-system
(\lean{isPiSystem\_predictableRect}) that generates the predictable
$\sigma$-algebra; the generation lemma itself
(\lean{ElementaryPredictableSet.generateFrom\_eq\_predictable}) is consumed
from \lean{BrownianMotion}. Take $\psi \in H_T$ orthogonal to every elementary
integral. Testing against the indicator of a rectangle $(a,b] \times F$ shows
that the integral of $\psi$ over that rectangle vanishes; the collection of
measurable sets on which $\int \psi = 0$ is a $\lambda$-system containing the
rectangle $\pi$-system, so Dynkin's $\pi$--$\lambda$ theorem extends the
vanishing to the entire predictable $\sigma$-algebra
(\lean{setIntegral\_eq\_zero\_of\_orthogonal\_pred}), forcing $\psi = 0$ a.e.
The It\^o integral $\ItoI_T = $ \lean{itoIntegralCLM\_T} is then the unique
isometric continuous-linear extension of the elementary integral, via
Mathlib's \lean{LinearIsometry}/\lean{extendOfNorm} idiom for extending a
norm-preserving map off a dense subspace.

\subsection{The keystone example}

A Riemann-sum bridge (\lean{riemannφ}, \lean{stepφ}, with cell-collapse
lemmas) connects $\ItoI_T$ applied to time-discretized integrands with
classical partial sums. Its first consumer is the closed form
\[
  \ItoI_T(B) \;=\; \tfrac12\,\big(B_T^2 - B_0^2 - T\big),
\]
the machine-checked $\int_0^T B\,dB$. The It\^o correction $-T/2$ emerges as
follows: the summation-by-parts identity \lean{discrete\_squaring\_identity}
writes
\[
  \sum_k B_{s_k}(B_{s_{k+1}} - B_{s_k}) = \tfrac12\Big(B_T^2 - B_0^2 -
  \textstyle\sum_k (B_{s_{k+1}} - B_{s_k})^2\Big),
\]
whose left side is the discretized $\int_0^T B\,dB$, converging in $L^2$ to
$\ItoI_T(B)$, while the squared-increment sum converges in $L^2$ to $T$
(Section~\ref{sec:formula}); what is left is $\tfrac12(B_T^2 - B_0^2 - T)$. We display it because it is the smallest statement in which the
calculus is visibly non-classical, because it later anchors the remainder
analysis of It\^o's formula, and because, as Section~\ref{sec:scope}
notes, it is an example the bounded-derivative It\^o formula cannot reach
on its own, which is why it is proved directly.

\subsection{Design note: two isometries}

The library separately maintains a Wiener-integral isometry for
\emph{deterministic} integrands. The two are deliberately distinct: the
Wiener integral needs no filtration and serves the Gaussian-process and
distributional layers, while $\ItoI_T$ is the adapted object that the
martingale theory of Section~\ref{sec:process} requires. Collapsing the
two would shorten the code and damage the mathematics; the duplication is
the design, and it is recorded as such in the dependency blueprint so that
a future reader does not ``simplify'' it away.

\section{The integral as a process}\label{sec:process}

\subsection{The simple layer}

For a simple process $V$ and $t \le T$, the partially-summed elementary
integral $(V \bullet B)_t$ is adapted
(\lean{itoSimpleProcess\_adaptedAt}) and a martingale
(\lean{itoSimpleProcess\_isMartingale}). The martingale step reduces to the
conditional vanishing of adapted-weighted future increments
(Listing~\ref{lst:incr}), proved in \lean{condExp\_adapted\_mul\_increment}
by upgrading the unconditional statement $\E[\xi(B_{t_1} - B_{t_0})] = 0$ to
its conditional form $\E[\xi(B_{t_1} - B_{t_0}) \mid \F_u] = 0$ for $u \le
t_0$. By the characterizing property of conditional expectation
(\lean{ae\_eq\_condExp\_of\_forall\_setIntegral\_eq}), it suffices that
$\int_A \xi(B_{t_1} - B_{t_0})\,d\mu = 0$ for every $A \in \F_u$; and $\xi
\mathbf{1}_A$ is again bounded and $\F_{t_0}$-adapted, so the increment
$B_{t_1} - B_{t_0}$, independent of $\F_{t_0}$ with mean zero, makes the
integral factor as $\E[\xi \mathbf{1}_A]\,\E[B_{t_1} - B_{t_0}] = 0$. The simple layer also
carries the per-time isometry $\E[(V\bullet B)_t^2] = \int_0^t \E[V_s^2]\,ds$
(\lean{itoSimpleProcess\_isometry\_time}) and $L^2$-continuity in $t$
(\lean{itoSimpleProcessLp\_l2\_continuous}); the next subsection lifts all
three to general integrands at once.

\begin{leanlisting}{The conditional martingale-difference step: adapted weights kill future increments.}{lst:incr}
theorem condExp_adapted_mul_increment (hBmeas : ∀ t, Measurable (B t))
    {t₀ t₁ : ℝ≥0} (ht : t₀ ≤ t₁) {φ : Ω → ℝ}
    (hφ : AdaptedAt B t₀ φ) {C : ℝ} (hC : ∀ ω, |φ ω| ≤ C) :
    μ[fun ω => φ ω * (B t₁ ω - B t₀ ω) | natFiltration hBmeas t₀] =ᵐ[μ] 0
\end{leanlisting}

\subsection{The key identity}\label{sec:keyidentity}

For general $\varphi \in H_T$ the process is defined by a continuous linear
map (\lean{itoProcessCLM}, obtained by extending the simple layer through
\lean{itoProcessLM.extendOfNorm}), and the development's central structural
fact is Theorem~\ref{thm:key} (Listing~\ref{lst:key}).

\begin{theorem}[{\lean{itoProcessCLM\_eq\_condExpL2}}]\label{thm:key}
For every $t \le T$ and $\varphi \in H_T$, as elements of
$L^2(\Omega,\mu)$,
\[
  (\varphi \bullet B)_t \;=\; \E\big[\, \ItoI_T(\varphi) \;\big|\; \F_t
  \,\big],
\]
the right side realized by Mathlib's $L^2$ conditional-expectation
projection \lean{condExpL2}.
\end{theorem}

\begin{leanlisting}{The key identity: the integral process is the conditional projection of its terminal value.}{lst:key}
theorem itoProcessCLM_eq_condExpL2 (T t : ℝ≥0) (hBmeas : ∀ u, Measurable (B u))
    (φ : Lp ℝ 2 (trimMeasure_T T hBmeas)) :
    itoProcessCLM T t hBmeas φ
      = (condExpL2 ℝ ℝ ((natFiltration hBmeas).le t)
          (itoIntegralCLM_T T hBmeas φ) : Lp ℝ 2 μ)
\end{leanlisting}

The proof is two lines of mathematics: both sides are continuous linear in
$\varphi$; on the dense simple subspace they agree, because there the
identity \emph{is} the simple layer's martingale property; conclude by the
equalizer property of dense ranges (\lean{DenseRange.equalizer}).
Everything else about the process is a corollary of Theorem~\ref{thm:key}
plus standard properties of conditional expectation:

\begin{itemize}
\item \emph{adaptedness}, since a \lean{condExpL2} projection lands in the
  $\F_t$-measurable subspace;
\item the \emph{martingale property}
  (\lean{itoIntegralProcessGen\_isMartingale}), by the conditional-expecta\-tion
  tower;
\item the \emph{contraction bound} $\|(\varphi\bullet B)_t\|_{L^2} \le
  \|\varphi\|_{H_T}$ (\lean{itoProcessCLM\_norm\_le}), since conditional
  expectation is a contraction;
\item the \emph{terminal isometry} at $t = T$
  (\lean{itoProcessCLM\_norm\_terminal});
\item \emph{$L^2$-continuity} of $t \mapsto (\varphi\bullet B)_t$
  (\lean{itoIntegralProcessGen\_l2\_continuous}), by a $t$-uniform
  contraction estimate against an approximating simple sequence
  (\lean{TendstoUniformly}), so that continuity transfers from the simple
  layer in the uniform limit.
\end{itemize}

Defining the process as a projection of its terminal value is what turns the
martingale theory of the It\^o integral into a sequence of corollaries rather
than a sequence of constructions.

\subsection{The time-indexed It\^o isometry}\label{sec:perT}

The terminal isometry above is the classical $\|\ItoI_T(\varphi)\|_{L^2}^2 =
\|\varphi\|_{H_T}^2$. Its refinement to intermediate times (the
energy law of the integral \emph{as a process}) is
Theorem~\ref{thm:isometry} (Listing~\ref{lst:iso}).

\begin{theorem}[{\lean{itoProcessCLM\_norm\_sq}}]\label{thm:isometry}
For every $t \le T$ and $\varphi \in H_T$,
\[
  \E\big[(\varphi \bullet B)_t^2\big] \;=\;
  \int_{(0,t]\times\Omega} \varphi^2 \; d\big((\lambda_T \otimes \mu)
  \!\restriction_{\mathcal{P}_T}\big) \;=\; \int_0^t \E[\varphi_s^2]\,ds .
\]
\end{theorem}

\begin{leanlisting}{The time-indexed It\^o isometry for general predictable integrands.}{lst:iso}
theorem itoProcessCLM_norm_sq {t T : ℝ≥0} (htT : t ≤ T)
    (hBmeas : ∀ u, Measurable (B u)) (φ : Lp ℝ 2 (trimMeasure_T T hBmeas)) :
    ‖itoProcessCLM T t hBmeas φ‖ ^ 2
      = ∫ z in (Set.Ioc 0 t ×ˢ (Set.univ : Set Ω)), (φ z) ^ 2
          ∂(trimMeasure_T T hBmeas)
\end{leanlisting}

The proof is again the equalizer pattern, against a different continuous
functional: $\varphi \mapsto \|(\varphi\bullet B)_t\|^2$ and
$\varphi \mapsto \int_{(0,t]} \varphi^2$ are both continuous, and they agree
on the dense simple processes by the band reconciliation
\lean{itoSimpleProcessLp\_band\_isometry}, which is where the trimmed
measure is decomposed across $(0,t] \cup (t,T]$, the supporting lemma being
\lean{integral\_rectTerm\_mul\_band}. The mathematical content is small once
the simple-layer isometry and the key identity are in place; we call it out
as a separate theorem because the second-moment law of the process is one of
the two facts every textbook labels ``the It\^o isometry,'' and stating it
for \emph{general} integrands at \emph{every} $t$ (not only at $t = T$,
and not only for simple processes) is what makes the integral-as-a-process
a complete object rather than a martingale with a missing variance law.

\section{It\^o's formula}\label{sec:formula}

\subsection{Architecture}

The route to It\^o's formula is discrete-to-continuous (Figure~\ref{fig:arch}):
\begin{enumerate}
\item \textbf{Discrete identity.} For any partition, second-order Taylor
  expansion of $f$ along the increments gives an exact discrete It\^o
  identity (\lean{DiscreteIto}) with an explicit Lagrange remainder.
\item \textbf{Weighted quadratic variation.} The $\sum f''(B_{t_k})
  (\Delta_k B)^2$ term converges in $L^2$ to $\int_0^T f''(B_s)\,ds$;
  this is proved for adapted weight processes
  (\lean{WeightedQuadraticVariation}), generalizing the unweighted $L^2$
  quadratic-variation limit and reusing the second- and fourth-moment
  Gaussian estimates of \lean{BrownianQuadraticVariation}.
\item \textbf{Remainder.} The third-order remainder is $O(1/n)$ in $L^2$
  under the $|f'''| \le C_3$ bound (\lean{ItoFormulaRemainder}); the
  time-dependent variant needs a two-dimensional remainder with mixed
  partials, also $O(1/n)$ (\lean{ItoFormulaTDRemainder}).
\item \textbf{Riemann bridge.} The discretized stochastic term is $\ItoI_T$
  applied to a step-process realization of $f'(B)$; the bridge lemmas
  identify partial sums with the CLM applied to \lean{stepφ}.
\item \textbf{Limit.} The step integrands converge in $H_T$ to a realization
  $g_{f'}$ of $s \mapsto f'(B_s)$. Predictability of the limit is free
  (the limit is taken \emph{inside} $H_T$, Section~\ref{sec:integral}); the
  existence of the realization for bounded continuous data is
  \lean{itoIntegralCLM\_T\_of\_bdd\_cont}, a dominated-convergence argument
  on the trimmed measure.
\end{enumerate}

\begin{figure}[t]
\centering
\begin{tikzpicture}[
  box/.style={draw, rounded corners, align=center, font=\scriptsize,
    text width=2.4cm, minimum height=8mm, inner sep=2pt},
  >={Stealth[width=4pt]}, x=1cm, y=1cm]
  \node[box] (rect) at (0,3)    {predictable $\pi$-system and density};
  \node[box] (clm)  at (4,3)    {$L^2$ It\^o integral $\ItoI_T$};
  \node[box] (proc) at (8,3)    {integral as a process};
  \node[box] (ito)  at (4,1.5)  {It\^o's formula};
  \node[box] (disc) at (0,0)    {discrete It\^o identity};
  \node[box] (qv)   at (3,0)    {QV in $L^2$ ($2t^2/n$)};
  \node[box] (wqv)  at (6,0)    {weighted QV in $L^2$};
  \node[box] (rem)  at (9,0)    {remainder $O(1/n)$};
  \draw[->] (rect) -- (clm);
  \draw[->] (clm) -- (proc);
  \draw[->] (qv) -- (wqv);
  \draw[->] (clm) -- (ito);
  \draw[->] (disc) -- (ito);
  \draw[->] (wqv) -- (ito);
  \draw[->] (rem) -- (ito);
\end{tikzpicture}
\caption{Proof architecture. The isometry layer (top) builds the integral and
the process; the discrete-to-continuous limit layer (bottom) supplies the three
vanishing terms; they meet at It\^o's formula.}
\label{fig:arch}
\end{figure}
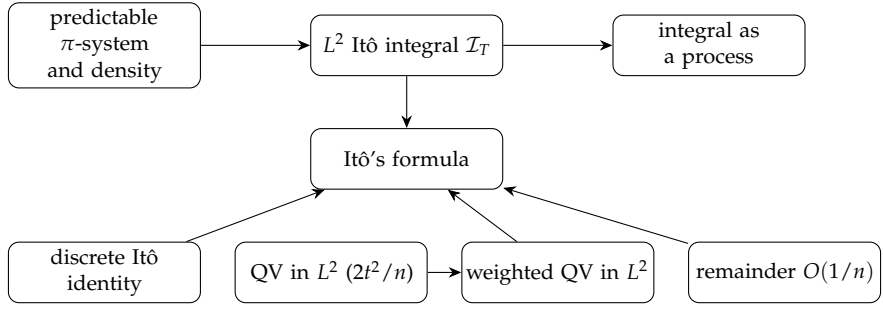

\subsection{The fluctuation engine: quadratic variation in $L^2$}

The second-order term of It\^o's formula is deterministic in the limit for one
reason: the squared increments of Brownian motion have vanishing fluctuation
in $L^2$. Along the uniform partition of $[0,t]$ into $n$ pieces of width
$\Delta = t/n$,
\[
  \Big\| \sum_k (B_{s_{k+1}} - B_{s_k})^2 - t \Big\|_{L^2(\mu)}^2
  = \frac{2t^2}{n} \longrightarrow 0
\]
(\lean{tendsto\_qv}), and the value is exact, not asymptotic. Writing $Y_k =
(B_{s_{k+1}} - B_{s_k})^2 - \Delta$ for the centered squared increment, the
$Y_k$ over disjoint intervals are independent by the weak Markov property of
$B$, so the cross terms vanish and $\E[(\sum_k Y_k)^2] = \sum_k \E[Y_k^2]$.
Each term is the Gaussian kurtosis,
\[
  \E[Y_k^2] = \E[(B_{s_{k+1}} - B_{s_k})^4] - \Delta^2 = 3\Delta^2 - \Delta^2
  = 2\Delta^2,
\]
using $\E[X^4] = 3(\operatorname{Var} X)^2$ for a centered Gaussian
(\lean{integral\_pow4\_gaussianReal}, transported to the increment law by
\lean{integral\_increment\_pow4}). Summing, $n \cdot 2\Delta^2 = 2t^2/n$. The
fourth-moment constant $3$, and nothing else, is the reason the quadratic
variation is $t$ rather than $0$; this is the formal content of $(\mathrm
dB)^2 = \mathrm dt$.

\subsection{The weighted second-order term}

It\^o's formula needs this limit carrying a weight: for a bounded adapted
process $w$ with continuous paths,
\[
  \sum_k w_{t_k}\,(B_{t_{k+1}} - B_{t_k})^2 \longrightarrow \int_0^T w_s\,ds
  \qquad \text{in } L^2
\]
(\lean{tendsto\_weighted\_qv\_process}), instantiated at $w_s = f''(B_s)$ for
the autonomous formula and $w_s = f_{xx}(s, B_s)$ for the time-dependent one.
The proof splits the difference into a fluctuation term and a Riemann term.
The fluctuation term $\sum_k w_{t_k}\big((B_{t_{k+1}} - B_{t_k})^2 -
\Delta_k\big)$ vanishes in $L^2$ by the orthogonality-plus-kurtosis estimate
above, now carrying the $\F_{t_k}$-measurable bounded weight: adaptedness of
$w_{t_k}$ is exactly what keeps the cross terms zero
(\lean{integral\_adapted\_mul\_centered\_sq}). The Riemann term $\sum_k
w_{t_k}\Delta_k - \int_0^T w_s\,ds$ vanishes pathwise by continuity of $s
\mapsto w_s(\omega)$, then in $L^2$ by dominated convergence against the bound
$|w| \le C$. Neither half uses where the weight came from, only that it is
adapted, bounded, and path-continuous, which is why one lemma serves both the
autonomous and the time-dependent formula.

\subsection{The remainder}

What survives after the first- and second-order terms is a sum of cubic Taylor
remainders, and it must vanish in $L^2$. The per-step remainder $R_k = f(B_{s_{k+1}})
- f(B_{s_k}) - f'(B_{s_k})(B_{s_{k+1}} - B_{s_k}) - \tfrac12 f''(B_{s_k})
(B_{s_{k+1}} - B_{s_k})^2$ obeys the cubic bound $|R_k| \le C_3\,|B_{s_{k+1}} -
B_{s_k}|^3$ under $|f'''| \le C_3$ (\lean{abs\_discreteTaylorRemainder\_le},
three applications of the convex mean-value inequality, each gaining one power
of the increment). Hence $\E[R_k^2] = O\big(\E[(B_{s_{k+1}} - B_{s_k})^6]\big)
= O(\Delta_k^3)$ by the Gaussian sixth moment
(\lean{integral\_pow6\_gaussianReal}), and Cauchy--Schwarz across the $n$ terms
gives
\[
  \E\Big[\Big(\textstyle\sum_k R_k\Big)^2\Big] \le n \sum_k \E[R_k^2]
  = O\big(n \cdot n \cdot (T/n)^3\big) = O(1/n) \longrightarrow 0
\]
(\lean{ItoFormulaRemainder}). The third order is the only place a third
derivative is used, and where boundedness of $f'''$ enters; the time-dependent
variant runs the same estimate on a two-dimensional Taylor remainder with
mixed partials (\lean{ItoFormulaTDRemainder}).

\subsection{The bounded-derivative It\^o formula}

\begin{theorem}[{\lean{ito\_formula\_L2\_bddDeriv}}]\label{thm:ito}
Let $B$ satisfy the pre-Brownian interface with measurable coordinates and
continuous paths, and let $f : \R \to \R$ be everywhere three-times
differentiable with $|f'| \le C_1$, $|f''| \le C_2$, $|f'''| \le C_3$. Then
there exists $g_{f'} \in H_T$, an $H_T$-realization of $s \mapsto f'(B_s)$,
such that $\mu$-a.e.
\[
  f(B_T) - f(B_0) \;=\; \ItoI_T(g_{f'}) \;+\; \tfrac12 \int_{(0,T]}
  f''(B_s)\,ds .
\]
\end{theorem}

\begin{leanlisting}{It\^o's formula, bounded-derivative class (hypotheses abbreviated).}{lst:ito}
theorem ito_formula_L2_bddDeriv (hBmeas : ∀ t, Measurable (B t))
    (hBcont : ∀ ω, Continuous (fun s => B s ω)) (T : ℝ≥0)
    {f f' f'' f''' : ℝ → ℝ}
    (hf : ∀ x, HasDerivAt f (f' x) x) (hf' : ∀ x, HasDerivAt f' (f'' x) x)
    (hf'' : ∀ x, HasDerivAt f'' (f''' x) x)
    (hf1 : ∀ x, |f' x| ≤ C1) (hf2 : ∀ x, |f'' x| ≤ C2) (hf3 : ∀ x, |f''' x| ≤ C3) :
    ∃ gf' : Lp ℝ 2 (trimMeasure_T T hBmeas),
      (fun ω => f (B T ω) - f (B 0 ω)) =ᵐ[μ]
        (fun ω => (itoIntegralCLM_T T hBmeas gf') ω
          + (1 / 2) * ∫ s in Set.Ioc 0 T, f'' (B s ω) ∂timeMeasure)
\end{leanlisting}

\emph{Proof sketch.} Fix the uniform partition of $[0,T]$ into $n$ pieces and
apply second-order Taylor expansion of $f$ to each increment. The telescoping
sum gives the exact discrete identity
\[
  f(B_T) - f(B_0) = S_n + \tfrac12 Q_n + R_n,
\]
with $S_n = \sum_k f'(B_{s_k})(B_{s_{k+1}} - B_{s_k})$ the discretized
stochastic integral, $Q_n = \sum_k f''(B_{s_k})(B_{s_{k+1}} - B_{s_k})^2$ the
weighted squared-increment sum, and $R_n$ the cubic remainder. Each term has a
limit established above: $S_n \to \ItoI_T(g_{f'})$ in $L^2$, because the step
integrands converge to $g_{f'}$ in $H_T$ and $\ItoI_T$ is isometric; $Q_n \to
\int_0^T f''(B_s)\,ds$ in $L^2$ by the weighted quadratic variation; and $R_n
\to 0$ in $L^2$ by the remainder estimate. Setting $a = S_n -
\ItoI_T(g_{f'})$, $b = \tfrac12(Q_n - \int_0^T f''(B_s)\,ds)$, and $c = R_n$,
the total error is $a + b + c$, and the elementary inequality $(a+b+c)^2 \le
3a^2 + 3b^2 + 3c^2$ drives $\E[(\text{total error})^2] \to 0$. The two sides of
the displayed identity therefore agree in $L^2$, hence $\mu$-almost everywhere.

Path continuity is supplied by the \lean{BrownianMotion} package's
continuous modification; it is genuinely used, as the Riemann limits run along
refining partitions evaluated on paths. The bounded-derivative hypotheses make
the $L^2$ analysis self-contained; their cost is discussed in
Section~\ref{sec:scope}.

\subsection{The time-dependent formula}

\begin{theorem}[{\lean{ito\_formula\_td\_L2\_bddDeriv}}]\label{thm:itoTD}
Let $f : \R \times \R \to \R$ admit everywhere partial derivatives
$f_t, f_x, f_{xx}$ together with $f_{tt}, f_{tx}, f_{xxx}$, all uniformly
bounded, with $f_x$ and $f_{xx}$ jointly continuous. Then there exists
$g_{f_x} \in H_T$ realizing $s \mapsto f_x(s, B_s)$ with, $\mu$-a.e.,
\[
  f(T, B_T) - f(0, B_0) \;=\; \ItoI_T(g_{f_x}) \;+\; \int_{(0,T]}
  \Big( f_t + \tfrac12 f_{xx} \Big)(s, B_s)\,ds .
\]
\end{theorem}

Two formalization points. First, joint continuity of $f_t$ is not assumed:
it is \emph{derived} from the bounded mixed partials
(\lean{continuous\_uncurry\_of\_bdd\_partials}), keeping the hypothesis
package at what a user can actually check. Second, the time-direction Taylor
terms are handled by a boundary-cancellation decomposition in which three
telescoping error sums vanish separately; the $L^2$ bound that glues the
stochastic and time directions is the three-term Cauchy--Schwarz estimate
visible in the remainder lemma's proof, where the squared error is dominated
by $3(\text{drift})^2 + \tfrac34(\text{QV})^2 + 3(\text{remainder})^2$ and
each term is sent to zero separately.

\subsection{Supporting results}

\paragraph{Quadratic variation.} For a process with measurable coordinates
and centered Gaussian increments of variance $t - s$ (\emph{neither
independence of increments nor path continuity is assumed}), the expected
squared-increment sums along equipartitions of $[0,t]$ converge to $t$
(\lean{qv\_equals\_t}); the $L^2$-limit form along refining partitions is in
\lean{QuadraticVariationL2}. That the hypothesis set is weaker than the
textbook's (the $L^1$ statement needs only the second moment of the
increments) is recorded in the statement, and it strictly strengthens the
classical claim.

\paragraph{Expectation-level It\^o.} For $f \in C^2$ with bounded
derivatives, $\E[f(B_t)] = f(0) + \tfrac12 \int_0^t \E[f''(B_s)]\,ds$
(\lean{expectation\_ito}; instance form
\lean{expectation\_ito\_isPreBrownian}). Historically the library's first
It\^o-flavored result, it now lives in the Feynman--Kac module and is the
bridge between the stochastic tower of this paper and the analytic
(heat-kernel) tower of the parent library.

\subsection{Notes on the proof engineering}

The density argument is the only place a monotone-class induction was unavoidable;
expressing it against a trimmed measure (so that the inductively-built sets
live in the right $\sigma$-algebra by construction) is what kept it
tractable. And the dominated-convergence step that produces the It\^o-formula
integrand realization (\lean{itoIntegralCLM\_T\_of\_bdd\_cont}) is delicate
precisely because the dominating function must itself be predictable; here too
the trimmed-measure design (rather than ambient $L^2$ with a side condition)
is what makes the domination land inside $H_T$ without a separate
measurability obligation.

\section{From $L^2$ to pathwise: the continuous local martingale}\label{sec:pathwise}

Sections~\ref{sec:process} and~\ref{sec:formula} are an $L^2$ theory: the
process $t \mapsto (\varphi\bullet B)_t$ is continuous \emph{as a curve in}
$L^2(\mu)$, and its martingale and isometry laws are identities of $L^2$
elements. The classical object is sharper. For almost every $\omega$ it has a
continuous \emph{path} $t \mapsto (\varphi\bullet B)_t(\omega)$, and it is a
continuous local martingale in the pathwise sense. This section builds that
object, first on $[0,T]$ and then on the whole half-line.

\subsection{The pathwise gap and the maximal inequality}

$L^2$-continuity in $t$ does not by itself give a single $\omega$-path that is
continuous: it controls each time separately, not the trajectory. The bridge
is a maximal inequality. The \lean{BrownianMotion} package supplies the $L^2$
maximal inequality for the integral process (\lean{maximal\_ineq\_norm});
Chebyshev turns it into a tail bound, a geometric subsequence makes the tails
summable, and Borel--Cantelli makes the approximating step-process integrals
converge uniformly on $[0,T]$ for almost every $\omega$. A uniform limit of
continuous functions is continuous, so the limit path is continuous off a null
set.

\subsection{The continuous modification on $[0,T]$}

Carrying that argument out gives an almost-surely continuous modification.

\begin{theorem}[{\lean{exists\_continuous\_modification\_itoProcess}}]\label{thm:mod}
For every $\varphi \in H_T$ there is a process
$\widetilde{X} : \R_{\ge 0} \times \Omega \to \R$ with $\mu$-a.e.\ continuous
paths such that $\widetilde{X}_t = (\varphi\bullet B)_t$ in $L^2(\mu)$ for
every $t \le T$.
\end{theorem}

The representative \lean{itoContinuousMod} is the pointwise limit
(\lean{limUnder}) of the approximating step integrals on the good set where
they converge uniformly. It equals $(\varphi\bullet B)_t$ a.e.\ at each $t$ by
uniqueness of limits in measure (\lean{tendstoInMeasure\_ae\_unique}), and its
paths are continuous by the uniform-convergence theorem
(\lean{TendstoUniformlyOn.continuousOn}); the supporting bound on the running
supremum over $[0,t]$ is \lean{norm\_le\_iSup\_Iic}.

\subsection{The local martingale on the null-augmented filtration}

A modification fixes the paths but can break adaptedness at the repaired
times, since the continuous representative is defined through a limit taken off
a null set. The remedy is to complete the filtration. We adjoin the $\mu$-null
sets to each $\F_t$, forming the null-augmented filtration
$\F_t^{\mu} = \F_t \vee \mathcal{N}$ (\lean{augFiltration}, the supremum of
\lean{natFiltration} with the constant null $\sigma$-algebra). Completeness is
all that is added; right-continuity, the other half of the usual conditions,
is neither added nor needed here.

\begin{theorem}[{\lean{exists\_continuous\_localMartingale\_modification}}]\label{thm:locmart}
For every $\varphi \in H_T$ there is a process $X$ with everywhere-continuous
paths such that $X_t = (\varphi\bullet B)_t$ in $L^2(\mu)$ for each $t \le T$,
and $X$ is a local martingale for $(\F^{\mu}_t)$.
\end{theorem}

Two points fix the precise content. First, the underlying $L^2$ process is a
\emph{true} martingale (Section~\ref{sec:process}); what
Theorem~\ref{thm:locmart} adds is its everywhere-continuous representative,
in the form the local-martingale interface of the \lean{BrownianMotion}
package takes, whose hypothesis of paths continuous for \emph{every} $\omega$
is what forces an everywhere-defined representative. Second, the martingale
property survives the completion: a conditional expectation is unchanged when
the filtration is enlarged by null sets (\lean{condExp\_sup\_nulls}), and this
is what carries adaptedness and the martingale identity from $(\F_t)$ to
$(\F^{\mu}_t)$.

\subsection{The unbounded horizon}

The bounded-horizon representatives are consistent: for $m \le n$ the
horizon-$m$ and horizon-$n$ continuous processes agree, a.e.\ pathwise, on
$[0,m)$, because both are continuous and both modify the same $L^2$ process
there (\lean{indistinguishable\_of\_modification\_on}). Gluing them gives a
single process on the whole half-line, reading at time $t$ the horizon
$\lceil t\rceil + 1$, which lies strictly beyond $t$.

\begin{theorem}[{\lean{exists\_continuous\_localMartingale\_modification\_infinite}}]\label{thm:inf}
For every $\varphi$ there is a process $X : \R_{\ge 0}\times\Omega \to \R$ with
everywhere-continuous paths such that $X_t = (\varphi\bullet B)_t$ in
$L^2(\mu)$ for \emph{every} $t$, and $X$ is a local martingale for the
null-augmented filtration $(\F^{\mu}_t)$.
\end{theorem}

Here there is no horizon clamp: the martingale property holds globally on
$\R_{\ge 0}$, supplied by the unbounded-horizon martingale identity through
\lean{condExp\_sup\_nulls}. This is the It\^o integral in the form the
classical theory gives it, a continuous local martingale on the half-line.

\section{Scope: what is, and is not, an It\^o calculus here}\label{sec:scope}

The title claims ``a'' It\^o calculus, and the indefinite article is doing
deliberate work. What earns the noun is that the canonical objects and the
canonical theorem are present and \emph{compose}: the stochastic term in
It\^o's formula (Theorem~\ref{thm:ito}) is definitionally the isometric
integral of Section~\ref{sec:integral}, whose energy law is the isometry of
Theorem~\ref{thm:isometry}, whose conditional projection is the martingale of
Theorem~\ref{thm:key}, whose engine is the quadratic-variation limit, and
whose everywhere-continuous representative is the local martingale of
Section~\ref{sec:pathwise}. One object, a single chain of theorems, no
re-derivation gluing them. That coherence is the sense in which this is a
calculus and not a collection of stochastic-analysis lemmas.

What the indefinite article concedes is equally concrete.

\paragraph{Bounded derivatives exclude the textbook's first examples.}
Theorem~\ref{thm:ito} requires $|f'| \le C_1$. The function $f(x) = x^2$ has
$f'(x) = 2x$, unbounded, and so is \emph{not} in the admissible class. In
particular the formula cannot, by itself, derive the keystone $\ItoI_T(B) =
\tfrac12(B_T^2 - B_0^2 - T)$, which is exactly why
Section~\ref{sec:integral} proves that identity directly, through the
discrete squaring route, rather than as an instance of It\^o's formula. The
honest reading is that what is formalized is the bounded-derivative It\^o
formula. The unrestricted $C^2$ statement is obtained, classically, by
localization (stopping the process at exit times of compact sets, applying
the bounded case, and passing to the limit), and so lies one standard
construction beyond the present development, atop the stopping-time layer
already available upstream.

\paragraph{Brownian integrator, $L^2$ integrands.} The integrator is Brownian
motion, not a general continuous local martingale or semimartingale, and
integrands live in the $L^2$ space $H_T$, not the localized class with
$\int_0^T \varphi_s^2\,ds < \infty$ only almost surely. Each is a standard way
to \emph{introduce} the It\^o integral (it is how the textbooks' first
chapters are scoped), and each is named here so that the verified frontier is
unambiguous. The $L^2$ integral and It\^o's formula are developed on a bounded
interval $[0,T]$; the pathwise local martingale of Section~\ref{sec:pathwise}
is the one object carried to the whole half-line.

\paragraph{Completeness, not the full usual conditions.} The pathwise local
martingale of Section~\ref{sec:pathwise} lives on the null-augmented
filtration: we add the $\mu$-null sets (completeness), which is what the
continuous modification needs, but not right-continuity, the other half of the
usual conditions. It is neither assumed nor used.

\paragraph{Quadratic variation in $L^2$, not pathwise.} Quadratic variation is
established in the $L^1$/$L^2$ sense, not in the pathwise (a.s.\ along refining
partitions) sense.

None of these boundaries qualifies a statement made in
Sections~\ref{sec:integral}--\ref{sec:pathwise}; they bound the \emph{reach}
of the development, not the \emph{validity} of any theorem in it.

\section{Related work}\label{sec:related}

Keskin's Isabelle/HOL martingale library~\cite{keskin2023} is the closest
prior art in the probabilistic direction; it develops (sub/super) martingales
in discrete and continuous time but stops below stochastic integration.
Degenne, Ledvinka, Marion, and Pfaffelhuber~\cite{brownianMotionRepo,degenne2025}
construct Brownian motion in Lean~4, building Kolmogorov extension, Gaussian
Banach-space measures, and Kolmogorov--Chentsov; that construction is the
foundation this work builds on, and the It\^o integral and its calculus are
the layer above it. The present development consumes their package and is, to
our knowledge (a search of the Isabelle AFP, the Lean ecosystem, and the Coq
ecosystem in June~2026), the first machine-checked It\^o calculus in any proof
assistant, and the first to carry the It\^o integral to a pathwise-continuous
local martingale. Avigad, H\"olzl, and Serafin's central limit
theorem~\cite{avigad2017} remains the historical anchor for deep
measure-theoretic probability in proof assistants. In the finance direction,
Echenim, Guiol, and Peltier formalized discrete-time replication pricing and
the Cox--Ross--Rubinstein model in Isabelle/HOL~\cite{echenim2020};
continuous-time mathematical finance above the present calculus is developed
in the parent library~\cite{flagship}. We are not aware of any other
formalization effort, in any assistant, that has published stochastic
integration, It\^o's formula, or their process-level martingale theory.

\section{Conclusion}\label{sec:conclusion}

The development formalizes the It\^o calculus of Brownian motion: on $[0,T]$,
the integral as an isometry, the integral as a continuous $L^2$ martingale
with its full second-moment structure, and It\^o's formula in both autonomous
and time-dependent forms, with the quadratic-variation machinery that drives
them; and then the pathwise object, an almost-surely continuous modification
whose everywhere-continuous representative is a local martingale on the whole
half-line. To our knowledge it is the first machine-checked It\^o calculus in
any proof assistant. It is also the layer the parent mathematical-finance
library~\cite{flagship} stands on, where the integral and the formula are
consumed to derive, rather than assume, the objects of continuous-time
pricing. The boundaries recorded in Section~\ref{sec:scope} (localization to
the unrestricted $C^2$ formula, integrators beyond Brownian motion, and
right-continuity of the filtration) mark where the formalized theory currently
ends; each sits on infrastructure that is the natural place for it to be
built.

\appendix
\section{Displayed results and their Lean names}\label{app:names}

\begin{table}[h]
\centering
\small
\resizebox{\textwidth}{!}{%
\begin{tabular}{@{}lll@{}}
\toprule
Result & Lean name & Module \\
\midrule
Integrand $\pi$-system & \lean{predictableRect} & \lean{ItoIntegralCLM} \\
\quad generator & \lean{generateFrom\_predictableRect} & \lean{ItoIntegralCLM} \\
Simple density & \lean{simpleAssembly\_T\_denseRange} & \lean{ItoIntegralCLM} \\
Elementary isometry & \lean{assembly\_isometry\_T} & \lean{ItoIntegralCLM} \\
The It\^o integral & \lean{itoIntegralCLM\_T} & \lean{ItoIntegralCLM} \\
Keystone $\int_0^T\! B\,dB$ & (\lean{ItoIntegralBrownian}) & \lean{ItoIntegralBrownian} \\
Cond.\ mart.\ difference & \lean{condExp\_adapted\_mul\_increment} & \lean{ItoIntegralProcessMartingale} \\
Key identity (Thm~\ref{thm:key}) & \lean{itoProcessCLM\_eq\_condExpL2} & \lean{ItoIntegralProcessGeneral} \\
Martingale property & \lean{itoIntegralProcessGen\_isMartingale} & \lean{ItoIntegralProcessGeneral} \\
$L^2$-continuity & \lean{itoIntegralProcessGen\_l2\_continuous} & \lean{ItoIntegralProcessGeneral} \\
Time-indexed isometry (Thm~\ref{thm:isometry}) & \lean{itoProcessCLM\_norm\_sq} & \lean{ItoIntegralProcessIsometry} \\
It\^o's formula (Thm~\ref{thm:ito}) & \lean{ito\_formula\_L2\_bddDeriv} & \lean{ItoFormulaCLM} \\
Time-dependent (Thm~\ref{thm:itoTD}) & \lean{ito\_formula\_td\_L2\_bddDeriv} & \lean{ItoFormulaTD} \\
Continuous modification (Thm~\ref{thm:mod}) & \lean{exists\_continuous\_modification\_itoProcess} & \lean{ItoIntegralProcessContinuousModification} \\
Local martingale (Thm~\ref{thm:locmart}) & \lean{exists\_continuous\_localMartingale\_modification} & \lean{ItoIntegralProcessLocalMartingaleGeneral} \\
$[0,\infty)$ local mart.\ (Thm~\ref{thm:inf}) & \lean{exists\_continuous\_localMartingale\_modification\_infinite} & \lean{ItoIntegralProcessLocalMartingaleInfinite} \\
Quadratic variation & \lean{qv\_equals\_t} & \lean{BrownianQuadraticVariation} \\
Expectation It\^o & \lean{expectation\_ito} & \lean{FeynmanKacHeatEquation} \\
\bottomrule
\end{tabular}%
}
\caption{Every displayed result, by Lean declaration name and module.}
\label{tab:names}
\end{table}

\end{document}